\documentclass[12pt]{article}

\usepackage[latin1]{inputenc}

\usepackage{fullpage}
\usepackage{amsfonts}
\usepackage{amssymb}
\usepackage{amsmath}

\usepackage[colors]{optsys}

\newcommand{\INDIVIDUAL}{\mathbb{I}}

\newcommand{\overlineINDIVIDUALS}[2]{\INDIVIDUAL\opt_{#1}\np{#2}}
\newcommand{\individual}{i}

\newcommand{\individualter}{\individual}
\newcommand{\source}{s}
\newcommand{\target}{t}
\newcommand{\EFFORT}[1]{\mathbb{X}_{#1}}         
\newcommand{\Effort}[1]{X_{#1}}         
\newcommand{\optEffort}[1]{X_{#1}\opt}
\newcommand{\effort}[2]{x_{#1#2}}        
\newcommand{\bareffort}[1]{\bar{x}_{#1\to}}         
\newcommand{\opteffort}[2]{x_{#1#2}\opt}
\newcommand{\investment}[1]{z_{#1}}
\newcommand{\multiplyerone}[1]{\lambda_{#1}}         
\newcommand{\multiplyertwo}[2]{\mu_{#1#2}}         
\newcommand{\relatedness}[2]{r_{#1#2}}
\newcommand{\fitness}[1]{F_{#1}}

\newcommand{\inclusivefitness}[1]{W_{#1}}
\newcommand{\BENEFICIARY}{\mathbb{B}}
\newcommand{\beneficiary}{\beta}
\newcommand{\BENEFICIARIES}[2]{\BENEFICIARY_{#1\to}\np{#2}}
\newcommand{\SELFISH}{\mathbb{S}}
\newcommand{\SELFISHES}[1]{\SELFISH\np{#1}}
\newcommand{\selfish}{\sigma}
\newcommand{\ALTRUISTIC}{\mathbb{A}}
\newcommand{\ALTRUISTICS}[1]{\ALTRUISTIC\np{#1}}
\newcommand{\altruistic}{\alpha}
\newcommand{\grandmother}{\textsc{\tiny GM}}
\newcommand{\grandfather}{\textsc{\tiny GF}}
\newcommand{\mother}{\textsc{\tiny M}}
\newcommand{\father}{\textsc{\tiny F}}
\newcommand{\offspring}{\textsc{\tiny O}}
\newcommand{\MGF}{\textrm{MGF}}

\title{Optimal Joint Allocation of Efforts\\ in Inclusive Fitness
by Related Individuals}

\author{Michel De Lara\\
CERMICS, Ecole des Ponts, Marne-la-Vall\'ee, France\\
michel.delara@enpc.fr}

\begin{document}

\maketitle

\begin{abstract}
  Families are places of affection and cooperation, but also of conflict.
  In his famous paper \emph{Parent-Offspring Conflict}, Robert L. Trivers builds
  upon W. D. Hamilton's concept of inclusive fitness to argue for genetic
  conflict in parent-offspring relationships, and to derive numerical
  predictions on the intensity of the conflict.
  We propose a mathematical model of game theory that depicts how each member of
  a family allocates her resource budget to maximize her inclusive
  fitness; this latter is made of the sum of personal fitness plus the sum of relatives fitnesses
  weighted by Wright's coefficients of relationship.
We define an optimal allocation profile as a Nash equilibrium, and we
characterize the solutions in function of 
 resource budgets, coefficients of relationship and derivatives of personal fitnesses. 
\end{abstract}






\section{Introduction}

  In his famous paper \cite{Trivers:1974}, Robert L. Trivers builds
  upon W. D. Hamilton's concept of \emph{inclusive fitness}
  \cite{Hamilton:1964} to derive numerical
  predictions on the intensity of the conflict in parent-offspring
  relationships. In this paper, we propose a mathematical analysis of 
optimal joint allocation of efforts in inclusive fitness 
by related individuals in a family. 

In Sect.~\ref{Problem_statement}, we set up the ingredients of a mathematical
model that depicts how each member of a family allocates her resource
budget to maximize her inclusive fitness;
we define an optimal allocation profile as a Nash equilibrium in game theory.
In Sect.~\ref{Optimal_allocation_of_efforts_by_related_individuals}, 
we characterize the solutions in function of  resource budgets, coefficients of
relationship and derivatives of personal fitnesses. 


\section{Problem statement}
\label{Problem_statement}

In~\S\ref{Investment_in_fitness,_relatedness,_inclusive_fitness},
we set up the mathematical ingredients to formalize 
investment in personal fitness, relatedness and inclusive fitness.
In~\S\ref{Selfish_and_altruistic_individuals_definitions}, we provide formal
definitions of selfish and altruistic individuals.
In~\S\ref{Solution_as_Nash_equilibrium}, 
we define an optimal allocation profile as a Nash equilibrium in game theory.

\subsection{Investment in personal fitness, relatedness, inclusive fitness}
\label{Investment_in_fitness,_relatedness,_inclusive_fitness}

Let \( \INDIVIDUAL \) be a finite set of individuals.
One can think of~\( \INDIVIDUAL \) as being relatives, members of a same family. 

\paragraph{Effort/investment.}
Each (source) individual~$\source \in \INDIVIDUAL$ has a budget~$\bareffort{\source} > 0 $ 
of effort/investment 
that she can allocate between all (target) individuals in~\( \INDIVIDUAL \),
including herself.\footnote{%
To avoid using she/he, her/him or herself/himself each time we refer to an
individual, we choose to use the feminine for a source individual
and the masculine  for a target individual.} 
Let \( \effort{\source}{\target} \) denote the 
\emph{effort/investment from~$\source$ to $\target$},
that is, the quantity that the source individual~$\source$ (first index) 
invests in the target individual~$\target$ (second index).
\begin{subequations}
Each (source) individual~$\source$ decides of an \emph{investment vector}
\begin{equation}
  \Effort{\source} = \sequence{ \effort{\source}{\target} }%
{ \target \in \INDIVIDUAL }
\in \RR_+^\INDIVIDUAL 
\eqsepv \forall \source \in \INDIVIDUAL 
\eqfinv 
\label{eq:investment_vector}
\end{equation}
having the property that 
its components are nonnegative ($\effort{\source}{\target} \geq 0$), 
and that the budget constraint is satisfied, that is, 
\begin{equation}
\sum_{\target\in \INDIVIDUAL} \effort{\source}{\target} \leq
\bareffort{\source}  
\eqsepv \forall \source \in \INDIVIDUAL 
\eqfinp 
\label{eq:budget_constraint}
\end{equation}
We note the \emph{set of admissible investment vectors 
for source individual~$\source$} by 
\begin{equation}
  \EFFORT{\source} = \defset{%
  \sequence{ \effort{\source}{\target} }%
{ \target \in \INDIVIDUAL }
\in \RR_+^\INDIVIDUAL }{%
\sum_{\target\in \INDIVIDUAL} \effort{\source}{\target} \leq
\bareffort{\source}  
}
\subset \RR_+^\INDIVIDUAL 
\eqsepv \forall \source \in \INDIVIDUAL 
\eqfinp 
\end{equation}
An \emph{admissible investment vectors profile} is a collection
\begin{equation}
\Effort{} = \sequence{ \Effort{\source} }{ \source \in \INDIVIDUAL }
= \sequence{ \effort{\source}{\target} }%
{ \np{\source,\target} \in \INDIVIDUAL^{2} }
\in \EFFORT{}
\label{eq:admissible_investment_vectors_profile}
\end{equation}
where
\begin{equation}
  \EFFORT{}=\prod_{\source \in \INDIVIDUAL} \EFFORT{\source}
  \subset \RR_+^{ \INDIVIDUAL \times \INDIVIDUAL }
\label{eq:set_of_admissible_investment_vectors_profiles}
\end{equation}
is the \emph{set of admissible investment vectors profiles}.
\end{subequations}

\paragraph{Personal fitness.}
For any admissible investment vectors profile~\( \Effort{} \in \EFFORT{} \) 
as in~\eqref{eq:admissible_investment_vectors_profile}--\eqref{eq:admissible_investment_vectors_profile}, 
we define the incoming investment profile 
\begin{subequations}
  \begin{equation}
 \Effort{\to} = \sequence{ \effort{\to}{\target} }{ \target \in \INDIVIDUAL }  
\in \RR_+^\INDIVIDUAL 
\eqfinv 
  \end{equation}
where 
\begin{equation}
\effort{\to}{\target} =
  \sum_{\source \in \INDIVIDUAL} \effort{\source}{\target}
\eqsepv \forall \target \in \INDIVIDUAL 
\label{eq:total_amount_invested_in_the_target_individual_b}
\end{equation}
is the \emph{total amount invested in the target individual~$\target$ 
by all source individuals~\( \source \in \INDIVIDUAL \),
including herself}.
\label{eq:total_amount_invested_in_the_target_individual}
\end{subequations}
By~\eqref{eq:budget_constraint}, we have that 
\begin{equation}
  \Effort{} \in \EFFORT{} \implies
\sum_{\target\in \INDIVIDUAL} \effort{\to}{\target} 
\leq \sum_{\source \in \INDIVIDUAL} \bareffort{\source}  
\eqfinp 
\end{equation}
We suppose that there exists a collection
\( \sequence{ \fitness{\target} }{ \target \in \INDIVIDUAL } \)
of \emph{personal fitness functions}
\begin{equation}
  \fitness{\target} : \RR_+ \to \RR_+ 
\eqsepv \forall \target \in \INDIVIDUAL 
\eqfinv
\label{eq:fitness}
\end{equation}
and that the \emph{personal fitness} of any target individual~$\target$
is \( \fitness{\target}\np{ \effort{\to}{\target} } \), that is,
personal fitness is a function 
of the total amount~\( \effort{\to}{\target} \) invested in~$\target$.

\paragraph{Coefficient of relatedness/relationship.}
The source individual~$\source$ shares with the target individual~$\target$
(including herself) 
a fraction \( \relatedness{\source}{\target} \in [0,1] \) of her genetic interests,
the \emph{coefficient of relatedness} (Wright's coefficient of relationship).
We suppose that
\begin{equation}
0 \leq \relatedness{\source}{\target} \leq
 \relatedness{\source}{\source} = 1 
\eqsepv \forall \np{\source, \target} \in \INDIVIDUAL^2
\eqfinp
\end{equation}

\paragraph{Inclusive fitness.}
We suppose that the \emph{inclusive fitness}~\(
\inclusivefitness{\individual} \) of the individual~$\individual$ 
depends on the whole admissible investment vectors profile
\( \Effort{}
= \sequence{ \effort{\source}{\target} }%
{ \np{\source,\target} \in \INDIVIDUAL^{2} }\)
in~\eqref{eq:admissible_investment_vectors_profile}--\eqref{eq:admissible_investment_vectors_profile}
in a way that is additive in the personal fitnesses of related individuals, 
after adjusting by relatedness, giving 
\begin{subequations}
  \begin{align}
\inclusivefitness{\individual}\np{ \Effort{} } 
&=
\sum_{\target\in \INDIVIDUAL} \relatedness{\individual}{\target} 
\fitness{\target}\np{ \effort{\to}{\target} }
\eqsepv \forall \individual \in \INDIVIDUAL 
\eqfinv
\label{eq:inclusive_fitness_a}
\intertext{where 
 the total amount~\( \effort{\to}{\target} \) invested in~$\target$
                                                  is given in~\eqref{eq:total_amount_invested_in_the_target_individual_b}.
                                                  Thus, we obtain the
                                                  equivalent, but more explicit,
                                                  expression of inclusive fitness}
\inclusivefitness{\individual} \bp{%
\sequence{ \effort{\source}{\target} }%
{ \np{\source,\target} \in \INDIVIDUAL^{2} } }
&=
\sum_{\target\in \INDIVIDUAL} \relatedness{\individual}{\target} 
\fitness{\target}\np{ \sum_{\source \in \INDIVIDUAL} \effort{\source}{\target} } 
\eqsepv \forall \individual \in \INDIVIDUAL 
\eqfinp 
\label{eq:inclusive_fitness_b}
  \end{align}
\label{eq:inclusive_fitness}
\end{subequations}

\subsection{Selfish and altruistic individuals definitions}
\label{Selfish_and_altruistic_individuals_definitions}

A (source) individual can have interest to invest in another (target) individual
because the source inclusive fitness ``includes'' the personal fitness of the target
individual as in~\eqref{eq:inclusive_fitness}.
A (target) individual can receive investment
from another (source) individual because the target
personal fitness is part of the source inclusive fitness  as in~\eqref{eq:inclusive_fitness}.

For a given admissible investment vectors profile 
\( \Effort{} 
= \sequence{ \effort{\source}{\target} }%
{ \np{\source,\target} \in \INDIVIDUAL^{2} }
\in \EFFORT{} \) as in~\eqref{eq:admissible_investment_vectors_profile}--\eqref{eq:admissible_investment_vectors_profile}, 
we say that 
\begin{itemize}
\item 
a (source) individual \( \selfish \in \INDIVIDUAL \) is \emph{selfish}
if 
\( \effort{\selfish}{\target} = 0 \)
for all \( \target \not = \selfish \),
that is, $\selfish$ does not invest 
in any other~$\target \not = \selfish $,
\item 
a (source) individual \( \altruistic \in \INDIVIDUAL \) is \emph{altruistic}
if 
\( \effort{\altruistic}{\target} > 0 \) for at least 
one other individual \( \target \not = \altruistic \),
\item 
a (source) individual~$\altruistic$ is \emph{totally altruistic} if
\( \effort{\altruistic}{\altruistic} = 0 \), that is, 
$\altruistic$ does not invest in herself, 
\end{itemize}
and we define 
\begin{subequations}
\begin{itemize}
\item 
the \emph{beneficiary target individuals} 
of the source individual~\( \source \in \INDIVIDUAL \) by
\begin{equation}
\BENEFICIARIES{\source}{\Effort{}} = 
\defset{ \beneficiary \in \INDIVIDUAL }%
{ \effort{\source}{\beneficiary} > 0 
}
\eqfinv
\label{eq:BENEFICIARIES}
\end{equation}
\item 
the \emph{selfish (source) individuals} as the individuals belonging to the set
\begin{equation}
\SELFISHES{\Effort{}} = \defset{ \selfish \in \INDIVIDUAL }%
{ 
\forall \target \not = \selfish 
\eqsepv
\effort{\selfish}{\target} = 0 
}
= 
\defset{ \selfish \in \INDIVIDUAL }%
{ 
\BENEFICIARIES{\selfish}{\Effort{}} = \{ \selfish \}
}
\eqfinv
\label{eq:SELFISHES}
\end{equation}
\item 
the \emph{altruistic (source) individuals} as the nonselfish individuals, that
is, those belonging to the set
\begin{equation}
\ALTRUISTICS{\Effort{}} = \INDIVIDUAL\backslash\SELFISHES{\Effort{}} 
= \defset{ \altruistic \in \INDIVIDUAL }%
{ 
\exists \target \not = \altruistic 
\eqsepv
\effort{\altruistic}{\target} > 0 
}
\eqfinp 
\end{equation}
\end{itemize}
\end{subequations}

\subsection{Solution as Nash equilibrium} 
\label{Solution_as_Nash_equilibrium} 

For a given admissible investment vectors profile 
\( \Effort{} = 
\sequence{ \Effort{\individual} }{ \individual \in \INDIVIDUAL } \in \EFFORT{}
\) as in~\eqref{eq:admissible_investment_vectors_profile}--\eqref{eq:admissible_investment_vectors_profile}, 
we denote
\begin{equation}
  \Effort{-\individual} = \sequence{ \Effort{\source} }%
{ \source \in \INDIVIDUAL, \source \neq \individual}
\in \prod_{\source \neq \individual} \EFFORT{\source} 
\eqsepv \forall \individual \in \INDIVIDUAL 
\eqfinp 
\end{equation}
%
%
A \emph{Nash equilibrium} is an admissible investment vectors profile
 \begin{subequations}
   \begin{equation}
     \optEffort{} = 
  \sequence{ \optEffort{\individual} }{ \individual \in \INDIVIDUAL }
= \sequence{  \sequence{ \opteffort{\individual}{\target} }%
{ \target \in \INDIVIDUAL }  }{ \individual \in \INDIVIDUAL }
 \in \EFFORT{} 
\eqfinv
\end{equation}
such that 
\begin{equation}
\max_{ \Effort{\individual} \in \EFFORT{\individual} } 
\inclusivefitness{\individual}\bp{
  \Effort{\individual} , \optEffort{-\individual} } = 
\inclusivefitness{\individual}\bp{
  \optEffort{\individual} , \optEffort{-\individual} } 
 \eqsepv 
\forall \individual \in \INDIVIDUAL 
\eqfinp
\label{eq:Nash_equilibrium}
\end{equation}
\end{subequations}



\section{Optimal allocation of efforts by related individuals}
\label{Optimal_allocation_of_efforts_by_related_individuals}

In~\S\ref{Characterization_of_Nash_equilibria}, we 
characterize Nash equilibria.
In~\S\ref{Selfish_and_altruistic_individuals_at_Nash_equilibrium},
we identify selfish and altruistic individuals at any Nash equilibrium.

\subsection{Characterization of Nash equilibria} 
\label{Characterization_of_Nash_equilibria}

\begin{proposition}
Suppose that the personal fitness 
functions~\( \sequence{ \fitness{\target} }{ \target \in \INDIVIDUAL } \)
in~\eqref{eq:fitness} are twice differentiable nondecreasing concave funtions,
that is, for all \( \target \in \INDIVIDUAL \),
\begin{itemize}
\item
\( \fitness{\target}' \geq 0 \), i.e. personal fitness is nondecreasing with
input investment,
\item
   \( \fitness{\target}'' \leq 0 \),
i.e. personal fitness displays nonincreasing returns with input
investment. 
\end{itemize}
Then, the investment vectors profile 
\( \optEffort{} 
= \sequence{ \opteffort{\source}{\target} }%
{ \np{\source,\target} \in \INDIVIDUAL^{2} }
 \in \EFFORT{} \)
is a Nash equilibrium (that is, a solution of~\eqref{eq:Nash_equilibrium}) if and only if 
\begin{subequations}
\begin{align}
\opteffort{\source}{\target} 
& \geq 0 
\eqsepv \forall \np{\source, \target} \in \INDIVIDUAL^2 
\eqfinv
\\
  \sum_{\source\in \INDIVIDUAL} \effort{\source}{\target}\opt 
&=
\effort{\to}{\target}\opt 
\eqsepv 
\forall \target \in \INDIVIDUAL 
\eqfinv
\\
\Bp{ \max_{\individualter \in \INDIVIDUAL} 
\relatedness{\source}{\individualter} 
\fitness{\individualter}'\np{ \effort{\to}{\individualter}\opt } 
- 
\relatedness{\source}{\target} 
\fitness{\target}'\np{ \effort{\to}{\target}\opt } }
\opteffort{\source}{\target} &= 0 
\eqsepv \forall \np{\source, \target} \in \INDIVIDUAL^2 
\eqfinv
\label{eq:KKTbis_complementarity_c}  
\\
\sum_{\target\in \INDIVIDUAL}
  \opteffort{\source}{\target} &=
\bareffort{\source} 
\eqsepv \forall \source \in \INDIVIDUAL 
\eqfinp
\label{eq:KKTbis_complementarity_b}  
\end{align}
\label{eq:KKTbis}  
\end{subequations}
\label{pr:Nash_equilibrium}
\end{proposition}

\begin{proof}
Each maximization problem~\eqref{eq:Nash_equilibrium}
consists of the maximization of the concave function 
\( \Effort{\individual} \in \EFFORT{\individual} \mapsto 
\inclusivefitness{\individual}\bp{
  \Effort{\individual} , \optEffort{-\individual} } 
\)
 in~\eqref{eq:inclusive_fitness}
over the convex domain~\( \EFFORT{\individual} \) 
defined by~\eqref{eq:budget_constraint}, giving 
\begin{subequations}
  \begin{align}
\bp{ \forall \individual \in \INDIVIDUAL }\quad
\max_{ \sequence{ \effort{\individual}{\target} }%
{ \target \in \INDIVIDUAL } }
&
\sum_{\target\in \INDIVIDUAL} \relatedness{\individual}{\target} 
\fitness{\target}\Bp{ \effort{\individual}{\target} +
\sum_{\source \neq \individual} \opteffort{\source}{\target} } 
\\
&  0 \leq \effort{\individual}{\target} \eqsepv 
\forall \target \in \INDIVIDUAL
\eqfinv
\\
& \sum_{\target\in \INDIVIDUAL} \effort{\individual}{\target} \leq
\bareffort{\individual}  
\eqfinp 
  \end{align}
\end{subequations}
This optimization problem
consists of the maximization of a concave differentiable function 
over a convex domain defined by affine inequalities. 
Therefore, the investment vectors profile 
\( \sequence{ \optEffort{\individual} }{ \individual \in \INDIVIDUAL } 
= \sequence{ \opteffort{\individual}{\target} }%
{ \np{\individual,\target} \in \INDIVIDUAL^{2} }
\)
is a Nash equilibrium (that is, solves~\eqref{eq:Nash_equilibrium})
if and only if
the following Karush-Kuhn-Tucker (KKT) conditions hold true 
\begin{subequations}
\begin{align}
0 
& \leq 
\opteffort{\source}{\target} 
\eqsepv \forall \np{\source, \target} \in \INDIVIDUAL^2 
\eqfinv
\\
\sum_{\target\in \INDIVIDUAL}
  \opteffort{\source}{\target} 
& \leq
\bareffort{\source} 
\eqsepv \forall \source \in \INDIVIDUAL 
\eqfinv 
\end{align}
and there exists \( \multiplyerone{}\opt= 
\sequence{ \multiplyerone{\source}\opt }{ \source \in \INDIVIDUAL } 
\in \RR^\INDIVIDUAL \) and 
\( \multiplyertwo{}{}\opt = \sequence{ 
\multiplyertwo{\source}{\target}\opt  }{ \np{\source, \target} \in \INDIVIDUAL^2
} \in \RR^{\INDIVIDUAL \times \INDIVIDUAL} \) which satisfy 
the following KKT conditions 
\begin{align}
  \relatedness{\source}{\target} 
\fitness{\target}'\np{ \effort{\to}{\target}\opt }
 &=  \multiplyerone{\source}\opt
- \multiplyertwo{\source}{\target}\opt
\eqsepv \forall \np{\source, \target} \in \INDIVIDUAL^2 \eqfinv
\label{eq:KKT_a}
\\
\multiplyertwo{\source}{\target}\opt \geq 0 & \eqsepv 
\multiplyertwo{\source}{\target}\opt 
\opteffort{\source}{\target} = 0 
\eqsepv \forall \np{\source, \target} \in \INDIVIDUAL^2 \eqfinv
\label{eq:KKT_complementarity_c}  
\\
\multiplyerone{\source}\opt \geq 0 & \eqsepv 
  \multiplyerone{\source}\opt \bp{ \sum_{\target\in \INDIVIDUAL}
  \opteffort{\source}{\target} -
\bareffort{\source} } = 0 
\eqsepv \forall \source \in \INDIVIDUAL 
\eqfinp
\label{eq:KKT_complementarity_b}  
\end{align}
\label{eq:KKT}  
\end{subequations}
We now show that Equations~\eqref{eq:KKT} are equivalent to
Equations~\eqref{eq:KKTbis}.

Suppose that Equations~\eqref{eq:KKT} hold true.
Let a (source) individual \( \source \in \INDIVIDUAL \) be fixed.
We first show that we necessarily have that 
  \begin{equation}
 \multiplyerone{\source}\opt = 
\max_{\target \in \INDIVIDUAL} 
\relatedness{\source}{\target} 
\fitness{\target}'\np{ \effort{\to}{\target}\opt } 
\eqfinp
\label{eq:inproof}
  \end{equation}
Indeed, from~\eqref{eq:KKT_a}
we deduce that 
\(  \multiplyerone{\source}\opt = 
\relatedness{\source}{\target} 
\fitness{\target}'\np{ \effort{\to}{\target}\opt } 
+ \multiplyertwo{\source}{\target}\opt \)
and, from \eqref{eq:KKT_complementarity_c}, that 
\( \multiplyerone{\source}\opt = 
\relatedness{\source}{\target} 
\fitness{\target}'\np{ \effort{\to}{\target}\opt } 
+ \multiplyertwo{\source}{\target}\opt 
\geq \relatedness{\source}{\target} 
\fitness{\target}'\np{ \effort{\to}{\target}\opt } \).
Therefore
\(  \multiplyerone{\source}\opt \geq 
\max_{\target \in \INDIVIDUAL} 
\relatedness{\source}{\target} 
\fitness{\target}'\np{ \effort{\to}{\target}\opt } 
\). 
It is impossible that the inequality be strict.
Indeed, else, we would have 
\( \multiplyerone{\source}\opt > 
\max_{\target \in \INDIVIDUAL} 
\relatedness{\source}{\target} 
\fitness{\target}'\np{ \effort{\to}{\target}\opt } \) 
and therefore, by~\eqref{eq:KKT_a}, 
\( \multiplyertwo{\source}{\target}\opt=
\multiplyerone{\source}\opt - 
 \relatedness{\source}{\target} 
\fitness{\target}'\np{ \effort{\to}{\target}\opt } > 0 \)
for all \( \target \in \INDIVIDUAL \).
Then, we would obtain that 
\( \opteffort{\source}{\target} = 0 \) 
for all \( \target \in \INDIVIDUAL \) 
by ~\eqref{eq:KKT_complementarity_c}.
We would also obtain that \( \multiplyerone{\source}\opt > 0 \),
as \( \multiplyerone{\source}\opt > 
\max_{\target \in \INDIVIDUAL} 
\relatedness{\source}{\target} 
\fitness{\target}'\np{ \effort{\to}{\target}\opt } 
\geq 0 \) 
since \( \fitness{\source}' \geq 0 \) by assumption. 
But \( \multiplyerone{\source}\opt > 0 \)
and \( \opteffort{\source}{\target} = 0 \) 
for all \( \target \in \INDIVIDUAL \) 
would contradict~\eqref{eq:KKT_complementarity_b},
since we have assumed that $\bareffort{\source} > 0 $.

Suppose that Equations~\eqref{eq:KKTbis} hold true.
To prove that Equations~\eqref{eq:KKT} hold true, we simply define~\( \multiplyerone{\source}\opt \)
by~\eqref{eq:inproof} and \( \multiplyertwo{\source}{\target}\opt \)
by~\eqref{eq:KKT_a}.
%
\end{proof}

\subsection{Selfish and altruistic individuals at Nash equilibrium}
\label{Selfish_and_altruistic_individuals_at_Nash_equilibrium}


For any vector \( \investment{} = 
\sequence{ \investment{\target} }{ \target \in \INDIVIDUAL } \in \RR_+^\INDIVIDUAL \) 
and for any source individual~\( \source \in \INDIVIDUAL \),
we introduce the set 
\begin{subequations}
  \begin{align}
\overlineINDIVIDUALS{\source}{\investment{}}   
&= 
\argmax_{\target \in \INDIVIDUAL}
\relatedness{\source}{\target} 
\fitness{\target}'\np{ \investment{\target} } 
\subset \INDIVIDUAL
\eqsepv \forall \source \in \INDIVIDUAL 
\label{eq:highest_adjusted_marginal_gain_in_fitness}
\intertext{of individuals with the \emph{highest adjusted marginal gain in
  personal fitness}, and the set}
  \overlineINDIVIDUALS{}{\investment{}}   
&= 
\argmax_{\individual \in \INDIVIDUAL }
\fitness{\individual}'\np{ \investment{\target} } 
 \subset \INDIVIDUAL
\label{eq:highest_marginal_gain_in_fitness}
      \end{align}
\end{subequations}
of individuals with the \emph{highest marginal gain in personal fitness}.

\begin{proposition}
  Suppose that the assumptions of Proposition~\ref{pr:Nash_equilibrium}
  hold true and that 
  the investment vectors profile 
\( \optEffort{} 
= \sequence{ \opteffort{\source}{\target} }%
{ \np{\source,\target} \in \INDIVIDUAL^{2} }
 \in \EFFORT{} \)
as in~\eqref{eq:admissible_investment_vectors_profile}--\eqref{eq:admissible_investment_vectors_profile}, 
is a Nash equilibrium (that is, a solution of~\eqref{eq:Nash_equilibrium}).

Then any individual \( \source \in \INDIVIDUAL \)
will only invest in those related individuals that have the
highest adjusted marginal gain in personal fitness, that is, 
\begin{equation}
\BENEFICIARIES{\source}{\optEffort{}} 
\subset 
\overlineINDIVIDUALS{\source}{\optEffort{\to}} 
\eqfinv
\end{equation}
where the beneficiary target individuals~\(
\BENEFICIARIES{\source}{\optEffort{}} \)
have been defined in~\eqref{eq:BENEFICIARIES},
individuals~\(\overlineINDIVIDUALS{\source}{\optEffort{\to}}\)
with the highest adjusted marginal gain in personal fitness 
in~\eqref{eq:highest_adjusted_marginal_gain_in_fitness}
and $\optEffort{\to}$
in~\eqref{eq:total_amount_invested_in_the_target_individual}.

Moreover, if \( \fitness{\target}'> 0 \) for any \( \target \in \INDIVIDUAL \), 
and if any individual is stricly more related to herself/himself
than to any other individual, that is, if 
\( \relatedness{\individual}{\individual} > \relatedness{\individual}{\target} \),
for any \( \individual \in \INDIVIDUAL \) and \( \target \neq \individual \),
then any individual \( \individual \in \INDIVIDUAL \)
that has the highest marginal gain in personal fitness will be selfish, that is, 
  \begin{equation}
  \overlineINDIVIDUALS{}{\optEffort{\to}}  \subset 
\SELFISHES{\optEffort{}}  
\eqfinv
\end{equation}
where individuals~\( \overlineINDIVIDUALS{}{\optEffort{\to}} \)
with the highest marginal gain in personal fitness 
have been defined in~\eqref{eq:highest_marginal_gain_in_fitness},
$\optEffort{\to}$
in~\eqref{eq:total_amount_invested_in_the_target_individual} 
and the selfish individuals~\( \SELFISHES{\optEffort{}} \)
in~\eqref{eq:SELFISHES}.
\end{proposition}

\begin{proof}
We have 
\begin{align*}
\target \in \BENEFICIARIES{\source}{\optEffort{}} 
 &\iff 
\opteffort{\source}{\target} > 0 
\tag{by definition~\eqref{eq:BENEFICIARIES} of $\BENEFICIARIES{\source}{\optEffort{}}$}
\\
&\implies   
\max_{\individualter \in \INDIVIDUAL} 
\relatedness{\source}{\individualter} 
\fitness{\individualter}'\np{ \effort{\to}{\individualter}\opt } 
- 
\relatedness{\source}{\target} 
\fitness{\target}'\np{ \effort{\to}{\target}\opt }
=0
\tag{by the KKT condition~\eqref{eq:KKTbis_complementarity_c}}
\\
 &\implies  
\target \in \overlineINDIVIDUALS{\source}{\optEffort{\to}} 
\tag{by definition~\eqref{eq:highest_adjusted_marginal_gain_in_fitness}
   of $\overlineINDIVIDUALS{\source}{\optEffort{\to}}$}
   \eqfinp 
\end{align*}

If  \( \fitness{\target}'> 0 \) for any \( \target \in \INDIVIDUAL \)
and if 
\( \relatedness{\individual}{\individual} > \relatedness{\individual}{\target} \),
for any \( \individual \in \INDIVIDUAL \) and \( \target \neq \individual \),
we have that 
\begin{align*}
\individual \in \overlineINDIVIDUALS{}{\optEffort{\to}} 
&\implies
\fitness{\individual}'\np{ \optEffort{\to} }
\geq 
\fitness{\target}'\np{ \optEffort{\to} }
\eqsepv \forall \target \neq \individual 
\tag{by definition~\eqref{eq:highest_marginal_gain_in_fitness}
of $\overlineINDIVIDUALS{}{\optEffort{\to}}$}
\\
&\implies
\fitness{\individual}'\np{ \optEffort{\to} }
\geq 
\fitness{\target}'\np{ \optEffort{\to} } > 0
\eqsepv \forall \target \neq \individual 
     \tag{by assumption that \( \fitness{\target}'> 0 \),
     \( \forall \target \) }
\\
&\implies
\relatedness{\individual}{\individual} 
\fitness{\individual}'\np{ \optEffort{\to} }
> 
\relatedness{\individual}{\target}
\fitness{\target}'\np{ \optEffort{\to} }
\eqsepv \forall \target \neq \individual 
     \tag{by assumption that \( \relatedness{\individual}{\individual} >
     \relatedness{\individual}{\target} \), \( \forall \target  \neq \individual \) }
\\
&\implies
\opteffort{\individual}{\target} =0
\eqsepv \forall \target \neq \individual 
\tag{by the KKT condition~\eqref{eq:KKTbis_complementarity_c}}
\\
&\implies
\individual \in \SELFISHES{\optEffort{\to}} 
\tag{by definition~\eqref{eq:SELFISHES} of selfish individuals}
     \eqfinp
\end{align*}

%
\end{proof}



\bibliographystyle{plainnat}
\bibliography{DeLara,psychology}

\end{document}